\documentclass{article}
\usepackage[in]{fullpage}
\usepackage{amsmath}
\usepackage{amsfonts}
\IfFileExists{pdfsync.sty}{\usepackage{pdfsync}}{}
\usepackage[cyr]{aeguill}

\usepackage{bibunits}

\newcommand\R{\mathbb{R}}

\newcommand\mS{\mathcal{S}}
\newcommand\mC{\mathcal{C}}

\newcommand\E[1]{\mathbb{E}[~#1~]}
\newcommand\Ec[2]{\mathbb{E}[~#1~|#2~]}

\newcommand\bb[1]{\gamma(#1)}

\renewcommand\u{c}
\newcommand\uu{u}

\newtheorem{definition}{Definition}
\newtheorem{lemma}{Lemma}
\newtheorem{theorem}{Theorem}
\newtheorem{corollary}{Corollary}
\newtheorem{proposition}{Proposition}
\newenvironment{proof}{{\bf Proof}:}{\hfill $\Box$ }

\newcommand\CONJECTURAL[1]{#1}
\newcommand\OPTIONAL[1]{{#1}}
\newcommand\OPTIMISTE[1]{#1}

\newcommand\equations[1]{
\begin{array}{lll}
#1
\end{array}
}

\newcommand\myaddress{ 
{\sc Ecole Polytechnique,} \\ {\sc Laboratoire LIX}\\ 91128 Palaiseau Cedex, \\FRANCE \\
{\small Olivier.Bournez@lix.polytechnique.fr}}
\newcommand\johanneaddress{ 
{\sc Centre National de la Recherche Scientifique,} \\{\sc Laboratoire PRiSM}\\ Universit\'e de Versailles \\
45, avenue des Etats-Unis, 78000 Versailles,\\ FRANCE \\
{\small Johanne.Cohen@prism.uvsq.fr}
}

\begin{document}

\title{Learning Equilibria in Games by Stochastic Distributed Algorithms.}

\author{
 Olivier Bournez \\[0.2cm]
\myaddress
\and
Johanne Cohen \\[0.2cm]
\johanneaddress
}
%\date{Submitted to SODA 2010\\ on  July 6th 2009.}
%\affiliation{\myaddress \and \johanneaddress}

\setcounter{page}{0}
\maketitle

\begin{abstract} We consider a class of fully stochastic and fully distributed algorithms, that we prove to learn equilibria in games.
  Indeed, we consider a family of stochastic distributed dynamics that we prove to converge weakly (in the sense of weak convergence for probabilistic processes) towards their mean-field limit, i.e an ordinary differential equation (ODE) in the general case.  We focus then on a class of stochastic dynamics where this ODE turns out to be related to multipopulation replicator dynamics. % , well-known and studied in evolutionary game theory.
  Using facts known about convergence of this ODE, we discuss the convergence of the initial stochastic dynamics: For general games, there might be non-convergence, but when convergence of the ODE holds, considered stochastic algorithms converge towards Nash equilibria. For games admitting Lyapunov functions, that we call Lyapunov games, the stochastic dynamics converge. We prove that any ordinal potential game, and hence any potential game is a Lyapunov game, with a multiaffine Lyapunov function. For Lyapunov games with a multiaffine Lyapunov function, we prove that this Lyapunov function is a super-martingale over the stochastic dynamics. This leads a way to provide bounds on their time of convergence by martingale arguments. This applies in particular for many classes of games that have been considered in literature, including several load balancing game scenarios and congestion games.
% We prove that for some classes of games that include congestion games, for which multipopulation replicator dynamics is known to be convergent, the corresponding stochastic algorithms can be proved to converge almost surely, and that it is possible by some martingale arguments to provide bounds on its time of convergence. 

% The considered dynamics are fully stochastic and fully distributed, and not based on best-response dynamics.
\end{abstract}

\newpage

\section{Introduction}

Consider a scenario where agents learn from their experiments, by small adjustments. This might  be  for example about choosing their telephone companies, or about their portfolio investments. %, that we will assume to be rational. % Assume that agents are rational.
We are interested in understanding when the whole market can converge towards rational situations, i.e. Nash equilibria in the sense of game theory. This is natural to expect dynamics of adjustments to be stochastic, and fully distributed, since we expect agents to adapt their strategies based on their local knowledge of the market, and since agents are often involved in games where a global, and hence local, deterministic description of the whole global market is not possible.  % We also want to avoid dynamics that would suppose a global knowledge of the market, that is to say we 
% We want to restrict to dynamics of adjustments that are fully distributed, in the sense that agents adapt their strategies based only on their local knowledge of the market, as a global knowledge of the market is often unknown. 

Several such dynamics of adjustments have been considered recently in the algorithmic game theory literature. Up to our knowledge, this has been done mainly for deterministic dynamics or best-response based dynamics: Computing a best response requires a global description of the market. Stochastic variations, avoiding a global description, have been considered. However, considered dynamics are somehow rather ad-hoc, in order to get efficient convergence time bounds, and still mainly best-response based. We want to consider here more general dynamics, and discuss when one may expect convergence. This could lead to consider any dynamics  which is monotone with respect to the utility of players, in relation with evolutionary game theory literature \cite{Evolutionary1}. We propose to restrict here to dynamics that lead to dynamics related to (possibly perturbed) replicator dynamics.

Somehow, as algorithmic game theory can be seen as an algorithmic version of classical game theory, our long term aim is to better understand algorithmic evolutionary game theory. Somehow, we could also say, that as best-response dynamics can be seen as strategies that visit corners of the simplex of (mixed) strategies, we are interested in a long term objective in learning methods that could be seen as interior point methods to find equilibria. 

%\section{Framework}

\vspace{0.1cm}
\textbf{Basic game theory framework.}  Let $[n]=\{1,\dots,n\}$ be the set of players. Every player $i$ has a set $\mS_i$ of \emph{pure strategies}. Let $m_i$ be the cardinal of $\mS_i$.  A \emph{mixed strategy} $q_i=(q_{i,1},q_{i,2},\dots,q_{i,m_1})$ corresponds to a probability distribution over pure strategies: pure strategy $\ell$ is chosen with probability $q_{i,\ell} \in [0,1]$, with $\sum_{\ell=1}^{m_i} q_{i,\ell}=1$. Let $K_i$ be the simplex of mixed strategies for player $i$.
Any pure strategy $\ell$ can be considered as mixed strategy $e_\ell$, where vector $e_\ell$ denotes the unit probability vector  with $\ell^{th}$ component unity, hence as a corner of $K_i$. 

Let $K=\prod_{i=1}^n K_i$ be the space of all mixed strategies. A \emph{strategy profile} $Q= (q_1, ..., q_n) \in K$ specifies the (mixed or pure) strategies of all players: $q_i$ corresponds to the mixed strategy played by player $i$.
Following classical convention, we write often write abusively $Q=(q_i,Q_{-i})$, where $Q_{-i}$ denotes the vector of the strategies played by all other players.

We allow games whose payoffs may be random: we only assume that whenever the strategy profile $Q \in K$ is known, each player $i$ gets a random \emph{cost} of expected value $c_i(Q)$. In particular, the expected cost for player $i$  for playing pure strategy $e_\ell$  is denoted by $c_i(e_\ell,Q_{-i})$.

% To simplify following discussions, let us assume that costs are non-negative: by linearity of expectation, then there must exist some constant $M$ with $c_i(Q) \in [0,M]$ for all $i$ and $Q$.

\vspace{0.1cm}
\textbf{Some classes of games.} Several classes of % (deterministic)
games where players' costs are based on
the shared usage of a common set of resources $[m]=\{1,2,\dots,m\}$ where each resource $1 \le r \le m$ has an associated nondecreasing cost
function denoted by $C_r : [n] \to \R$, have been considered in algorithmic game theory literature.

%\begin{definition}[Load Balancing Games \cite{KP99}]
  In \emph{load balancing games} \cite{KP99}, resources are called machines, and players compete for elements (i.e. singleton subsets) of $[m]$. Hence, the pure strategy space $\mS_i$ of player $i$
having a weight $w_i$ corresponds to $[m]$ or a subset of $[m]$, and a pure strategy $q_i \in \mS_i$ for player $i$ is some element $r \in [m]$.  
The cost for player (task) $i$  under profile of pure strategies (assignment) $Q=(q_1,\dots,
q_n)$  corresponds to $c_i(Q) = C_{q_i}(\lambda_{q_i}(Q))$, where $\lambda_r(Q)$ is the load of machine $r$:  $\lambda_r(Q) =
\sum_{j: q_j=r} w_j$, that is to say %  defined as 
the sum of the weights of the tasks running on it.

%
% The load $\lambda_r$ of machine $r$ under profile of pure strategies (assignment) $Q=(q_1,\dots,
% q_n)$ is
% the sum of the weights of the tasks running on it: $\lambda_r(Q) =
% \sum_{j: q_j=r} w_j$. The cost for player (task) $i$ then corresponds to the cost
% on machine $q_i$, i.e., its cost is $c_i = C_{q_i}(\lambda_{q_i}(Q))$.%  where
% % the $C_\ell(\lambda)$ are some non-decreasing function. 
% %\end{definition}

% This is a restricted instance of the more general congestion games. 
%\begin{definition}[Congestion Games ] 
In \emph{congestion games} \cite{rosenthal73}, resources are called edges, and players compete for subsets of $[m]$. \label{def:congestion}
Hence, the pure strategy space $\mS_i$ of player $i$ is a subset of $2^{[m]}$
and a pure strategy $q_i \in Q$ for player $i$ is a subset of $[m]$.  The cost of player $i$ under 
profile of pure strategies $Q$ corresponds to $c_i (Q) = \sum_{r \in q_i} C_r (\lambda_r (Q))$ where
$\lambda_r(Q)$ is the number of % players  that use resource $r$ in $Q$, that is to say the number of
$q_j$ with $r \in q_j$. In \emph{weighted congestion games}, weights $(w_i)_i$ are associated to players, and one takes instead $\lambda_r(Q)= \sum_{j:r\in q_j} w_j$.

In \emph{task allocation games} \cite{ChristodoulouKN04},  as in load balancing games, resources are called machines, and players compete for elements (i.e. singleton subsets) of $[m]$. 
% Hence, the pure strategy space $\mS_i$ of player $i$
% having a weight $w_i$ corresponds to $[m]$ or a subset of $[m]$, and a pure strategy $q_i \in \mS_i$ for player $i$ is some element $1 \le r \le m$.
Each resource (machine) $r$ is  assumed to have a function  $C_r$ that takes as input a set of tasks $\lambda \subset [n]$ assigned to it, and outputs a cost $C_{r,j}$ for each participating player $j$. The cost of player $i$ under 
profile of pure strategies $Q$ is then given by $c_i (Q) = C_{q_i,i}(\{j| q_j=q_i\})$. 
Functions $C_r$ can be considered as speed and scheduling policies, and associated costs as corresponding completion time for player (task) $i$.  For example, SPT and LPT are policies that schedule the jobs
without preemption respectively in order of increasing or decreasing weights (processing times) \cite{ChristodoulouKN04}.

% Each resource (machine) $r$ is assumed to have a speed and a policy that specifies how the players
% (tasks) associated to $i$ are scheduled. For example, SPT and LPT are policies that schedule the jobs
% without preemption respectively in order of increasing or decreasing
% processing times. 
Clearly, load balancing games are particular task allocation games, and load balancing games are particular weighted congestion games. A load balancing game whose weights are unitary is a particular congestion game. 

% q_n)$ be a profile of pure strategies.  The cost of player $i$, in a
% profile $Q$ is $c_i (Q) = \sum_{r\in q_i} d_r (n_r (Q))$ where
% $n_r(Q)$ is the number of players that use resource $r$ in $Q$, that is to say the number of $q_j$ with $r \in q_j$. 
% \end{definition}

% We can also consider the following.

% \begin{definition}[Allocation Games \cite{?}]
% In \emph{allocation games}, \dots
% \end{definition}

% In \emph{weighted} congestion game, player $i$ has weighted demand $w_i$.  Moreover, the congestion (load) $\ell_r(Q)$ on resource $r$ in a profile $Q$ is $\ell_r(Q)= \sum_{i:r\in q_i} w_i$.  % The cost of player
% % $i$ in $Q$ is now $c_i (Q) = \sum_{r\in q_i} d_r (\ell_r (Q ))$

%  In
% this work, we will consider the weighted congestion game with linear
% latency functions. Note that in this game, there is at least one pure
% Nash equilibrium (proved by using the potential function argument)

% I
% We will focus  in this paper on linear (affine) cost functions:
% $C_\ell(\lambda)= \alpha_\ell \lambda + \beta_\ell$, with $\alpha_\ell,\beta_\ell \ge
% 0$. This is sometimes called the \textit{uniformly related machine}
% case \cite{ChapterVocking}.

\vspace{0.1cm}
\textbf{Ordinal and potential games.} All these classes of games can be related to ordinal and potential games introduced by \cite{MondererShapley1996}: 
%
% All these games are specific instances of ordinal potential games.
%
% \begin{definition}[Ordinal and Potential Game \cite{MondererShapley1996}] \label{def:ordinal}.
A game is an \emph{ordinal potential game} if there exists some function $\phi$ from
  \emph{pure} strategies to $\R$ such that for all pure strategies
  $Q_{-i}$, $q_i$, and $q'_i$, one has 
$\u_i(q_i,Q_{-i})-\u_i(q'_i,Q_{-i}) >0 \mbox{ iff } \phi(q_i,Q_{-i}) - \phi(q'_i,Q_{-i})>0$.
It is an \emph{an (exact) potential game} \label{def:potential} if % there exists some function $\phi$ from \emph{pure} strategies to $\R$ such that
for all pure strategies $Q_{-i}$, $q_i$, and $q'_i$, one has
$\u_i(q_i,Q_{-i})-\u_i(q'_i,Q_{-i}) = \phi(q_i,Q_{-i}) - \phi(q'_i,Q_{-i}).$
% \end{itemize}
% \end{definition}

\section{Stochastic Learning Algorithms}

\textbf{Generic Stochastic Learning Algorithm.}
We want basically to consider learning algorithms of the following form, over the most possible general games, where
$b$ is a parameter, intended to be positive but close to $0$. 
\ \\

\def\td{\quad \quad}
\def\tt{\quad \quad \quad}

%\begin{figure} \label{fig:basic}
\hspace{-0.7cm}
\frame{
\begin{minipage}{\textwidth}
$\bullet$ Initially, $q_i(0) \in K_i$ can be any vector of
    probability, for all $i$. \\
$\bullet$  At each round $t$, 

\td $\bullet$ Any player $i$: selects a strategy $s_i(t) \in \mS_i$ according to
  distribution $q_i(t)$: player $i$ selects strategy $\ell \in \mS_i$
  with probability
  $q_{i,\ell}(t)$.   This leads to a (random) cost $r_i(t)$ for player $i$. 

\td $\bullet$  Select some player $i(t)$ at random: player $i$ is selected with probability
      $p_i$, with $\sum_{i=1}^n p_i=1$.

\tt This player $i=i(t)$ updates $q_i(t)$ as follows: %\begin{equation} \label{eq:dynamic}
%         q_i(t+1) = q_i(t) + b \frac{M- r_i(t)}{M} (e_{s_i(t)}-q_i(t))
 $      q_i(t+1) = q_i(t) + b F_i^b(r_i(t),s_i(t),q_i(t));
  $%     \end{equation}

\tt Any other player keeps $q_i(t)$ unchanged: $q_i(t+1) =q_i(t)$.
  \end{minipage}
  }

\vspace{0.1cm}

In a first step, consider functions $F_i^b(r_i(t),s_i(t),q_i(t)))$  as generic as possible, maintaining that the $q_i(t)$ always stay validity probability vectors: that is to say, $q_{i,\ell}(t) \in [0,1]$ and $\sum_{\ell} q_{i,\ell}(t)=1$ is preserved. Functions $F_i^b(r_i(t),s_i(t),q_i(t))$ can be random (formally a random variable). We only assume that its expectation $\Ec{F_i^b(r_i(t),s_i(t),q_i(t))}{Q(t)}$ is always defined.

This corresponds indeed to fully distributed algorithms\footnote{We of course understand that for some games (like congestion games), the size of the involved probability vectors might be non-polynomial. However, by restricting to function $F_i^b(r_i(t),s_i(t),q_i(t)))$, or close dynamics, which guarantee a support of polynomial size for $q_i(t)$, can solve the problem: restrict to function which are equal to $-bq_{i,\ell}$ for components $\ell$ outsides a polynomial (or fixed) sized support, for example.  If this is too problematic to our reader, please consider that we restrict to games where the $m_i$ stay polynomial, as for load balancing games and task allocation games.}. Decisions made by players are completely decentralized: At each time step, player $i$ only needs $r_i$ and $q_i$, that is to say respectively her cost and her current mixed strategy, to update his own strategy $q_i$. % We assume all the $0 <q_i <1$.

Let $Q(t) = (q_1(t), ..., q_n(t)) \in K$ denote the state of all players
at instant $t$. Our interest is in the asymptotic behavior of $Q(t)$, and its possible convergence to Nash equilibria. 
Assume that $G_i(Q)=\lim_{b \to 0} \Ec{F_i^b(r_i(t),s_i(t),q_i(t))}{Q}$ exists and is some continuous function $G_i$ of $Q$.

\vspace{0.2cm}
\textbf{Results.}
In the general case (Theorem \ref{propindiens}), any stochastic algorithm in the considered class converges% \COMMENTAPPENDIX{informalanalysis} 
weakly (in the sense of weak convergence for probabilistic processes) towards solutions of initial value problem (ordinary differential equation (ODE)) 
%\begin{equation} \label{eq:informal}
% $$
$\frac{dq_i}{dt}  =p_i G_i(Q),$ %$
%\end{equation}
 given $Q(0)$, i.e. to its mean-field limit approximation.

This can be seen informally as follows: Assume we replace $\Ec{\Delta
  q_i(t)}{Q(t)}$ by $\Delta q_i(t)$ in $\Ec{\Delta q_i(t)}{Q(t)} = b p_i \tilde{F}_i^b(Q(t))$, in the discussion that follows the description of the algorithm, where $\tilde{F}_i^b(Q(t))=\Ec{F_i^b(r_i(t),s_i(t),q_i(t))}{Q(t)}$.

Through the change of variable $t \leftarrow t b$, this would become 
$q _{i}(t+b)-q_{i}(t) = bp_i \tilde{F}_i^b(Q).$
Approximating $q_i(t+b)-q_i(t)$ by $b \frac{dq_i}{dt}(t)$ for small $b$, we may expect
the system to behave like ordinary differential equation (ODE)
\begin{equation} \label{eq:informal}
% $$
\frac{dq_i}{dt}  =p_i G_i(Q),%$$
\end{equation}
when $b$ is close to $0$. 

%
%\begin{definition}[Replicator-Like Dynamics]
A \emph{replicator-like dynamics} $F_i^b$ is a dynamic where \label{def:replicatorlike}
%\begin{equation} %\label{eq:dynamge}
$$  F_i^b(r_i(t),s_i(t),q_i(t))=  \gamma(r_i(t)) 
%\frac{M- r_i(t)}{M} 
(e_{s_i(t)}-q_i(t)) + \mathcal{O}(b), $$
%\end{equation}
or where this holds for its expectation,
where $\gamma:\R \to [0,1]$ is some \emph{decreasing}\footnote{ If we assume all costs to be positive, by linearity of expectation then all costs must be bounded by some constant $M$, and we can take for example $\gamma(x)=\frac{M- x}{M}$. 
} function with value in $[0,1]$.
 Recall that $e_{s_i(t)}$ is the unit vector of dimension $m_i$ with component number $s_i(t)$ unity.
%\end{definition}

Notice that we allow perturbed dynamics: $\mathcal{O}(b)$ denotes some perturbation that stay of order of parameter $b$.

We can also allow randomly perturbed dynamics:  a \emph{perturbed replicator-like dynamic} is of the form
$$ \label{per:urgence}
 F_i^b(r_i(t),s_i(t),q_i(t))= \mathcal{O}(b)+ 
\left\{
\begin{array}{ll}
% \frac{M- r_i(t)}{M}
\gamma(r_i(t))
(e_{s_i(t)}-q_i(t)) & \mbox{ with probability $\alpha$} \\
b (e_{s_j}-q_i(t)) & \mbox{ with probability } 1-\alpha, \\
& \mbox{  where } j\in \{1,\dots,m_i\} \\
& \mbox{ is chosen uniformly}, 
\end{array}
\right. %+ 
$$ 
where $0<\alpha<1$ is some constant.

% Indeed, for such dynamics equation \eqref{eqgene} leads to (a rescaling of) (multipopulation) classical replicator dynamics, 
% % \begin{equation}
% % \frac{dq_{i,\ell}}{dt} = -p_i \frac{\alpha}{M} q_{i,\ell} ( \u_i(e_{\ell},Q_{-i}) - \u_i(q_i,Q_{-i})).
% % \end{equation}
% whose limit points are well-known to be related to Nash equilibria.

We claim that such dynamics have a mean-field approximation which is isomorphic to a multipopulation replicator dynamics.% , whose limit points are known to be related to Nash equilibria.

We claim (Theorem \ref{th:deux}), that for general games, \emph{if}  there is convergence of the mean-field approximation, then stable limit points will correspond to Nash equilibria of the game. Notice, that there is no reason that convergence of mean-field approximation holds for generic games, but if it holds, then its stable limit points will be Nash equilibria.

% As convergence for continuous ODE is often obtained through Lyapunov function arguments, we propose to
%A MLgames whose continuous dynamics admits some Lyapunov function.

We claim (Theorem \ref{th:ordinal}) that ordinal games (and hence (exact) potential games) are \emph{Lyapunov} games: their mean-field limit approximation admits some  Lyapunov function. Furthermore, this Lyapunov function, that can be taken as the expectation of the potential and is of a special type, that we call \emph{multiaffine}.

We show that for Lyapunov games with multiaffine Lyapunov function (hence this includes ordinal and (exact) potential games such as load balancing, task allocation and congestion games), the Lyapunov function is a super-martingale over stochastic dynamics.

We deduce results on the convergence of stochastic algorithms for this class. % We claim that almost-surely after some time perturbed algorithms will be forever close to some Nash equilibria. 
We claim (Theorem \ref{th:timeone}) that for generic Lyapunov games with multiaffine Lyapunov function, the convergence towards Nash equilibria happens in expected time of order $ \frac{F(Q(0))}{\epsilon},$ taking $b$ of order $\epsilon$.

\vspace{0.1cm}
\textbf{Related work.}%\COMMENTAPPENDIX{results}
This is clear that an (exact) potential game is an ordinal potential game.  Congestion games, and hence load balancing games are known to be particular (exact) potential games \cite{rosenthal73}. Actually, it is known that a game is an (exact) potential game iff its is isomorphic to a congestion game \cite{MondererShapley1996}.  It has been proved in \cite{ChristodoulouKN04} that task allocation games are ordinal potential games, for SPT and LPT policies: it is proved that one can build some function $\phi$, which takes values of the form $(l_1,\cdots,l_n)$, that is lexicographically decreasing iff a player is doing a best response move.  As the $l_i$ (which corresponds to loads) are bounded by some constant $K$, function $\phi=\sum_{i} l_i K^{n-i}$ is decreasing iff a player is doing a best response move.

In other words, task allocation games under SPT and LPT policies  are indeed ordinal potential games, under the terminology of \cite{MondererShapley1996}.% \COMMENTAPPENDIX{com:ordinal}.

An ordinal potential game always have a pure Nash equilibrium: since 
ordinal potential function, that can take only a finite number of values, is strictly decreasing in any sequence of pure strategies strict best response moves, such a sequence must be finite and must lead to a Nash equilibrium \cite{rosenthal73}. This proof of existence of pure Nash equilibria can be turned into a dynamic: players play in turn, and move to resources with a lower cost. % Such a strategy can be proved to lead to a pure Nash equilibrium. 

For load-balancing games, following this idea, bounds on the convergence time of best-response dynamics have been investigated in \cite{% EvenDaretal03,
  Even-DarKM07}. % , using the hypothesis that each player has a global knowledge of the system in order to make its decision.
%
%The system reaches a Nash equilibrium after expected $O(\log \log m + \log n)$ rounds.  
% Obtained bounds are mostly obtained by bounding the possible variations of potential functions, using above argument.
Since players play in turns, this is often called the
\textit{Elementary Stepwise System}. Other results of convergence in
this model, have been investigated
in~\cite{GoldbergPODC2004,libman2001ars,orda1993crm}, but all require some global knowledge of the system in order to determine what next move to choose. 

A Stochastic version of best-response dynamics has been investigated in \cite{SODA06,BerenbrinkS07}. It is proved to terminate in expected $O(\log \log n + m^4)$ rounds for uniform tasks, and uniform machines. This has been extended to weighted tasks and uniform machines in \cite{BerenbrinkS07}. The expected time of convergence to an $\epsilon$-Nash equilibrium is in $\mathcal{O}(nmW^3 \epsilon^{œôø²2} )$ where $W$ denotes the maximum weight of any task.

For congestion games, the problem of finding pure Nash equilibria in congestion games is PLS-complete \cite{PLS}.  Efficient convergence of particular best-response dynamics to approximate Nash equilibria in symmetric congestion games have been  investigated  in \cite{ChienS07}, in the particular case where each resource cost function satisfies a \emph{bounded jump assumption}. In this context, the convergence to $\epsilon$-Nash equilibria occurs within a number of steps that is polynomial in the number of players.  This has been extended  to different classes of asymmetric congestion games in \cite{AwerbuchAEMS08}.

All previous discussions are about best-response dynamics. A stochastic dynamic, not elementary stepwise like ours, but close to those considered in this paper, has been partially investigated in \cite{DecentralizedLearning94} for general games and for potential games:  It is proved to be weakly convergent to solutions of a multipopulation replicator equation. Some of our arguments follow theirs, but notice that their convergence result (theorem 3.1) is incorrect: convergence may happen towards non-Nash (unstable) stationary points. Furthermore, this is not clear that any super-martingale argument holds for such dynamics, as our proof relies on the fact that the dynamics is elementary stepwise.

% Notice that compared to this paper, we also investigate more general classes of dynamics, and that no time bound discussion is done in this paper. 

%  to some function solution of an ordinary differential
% equation. This ordinary differential equation turns out to be a
% replicator equation. A sufficient condition for convergence is
% proved. No error bounds are provided in this paper. No Liapunov
% function is established for systems similar to the ones considered in
% this paper.

Replicator equations have been deeply studied in evolutionary game theory \cite{Evolutionary1,LivreWeibull}.  Evolutionary game theory has been applied to routing problems in the Wardrop traffic model in \cite{fischer2004esr,fischer2006fcw}. Potential games have been generalized to continuous player sets in \cite{sandholm2001pgc}. They have be shown to lead to multipopulation replicator equations, and since our dynamics are not about continuous player sets, but lead to similar dynamics, we borrow several % reasonings and
constructions from \cite{sandholm2001pgc}.  No time convergence discussion is done in \cite{sandholm2001pgc}.

A replicator equation for routing games has been considered in \cite{Altman2007}, where a Lyapunov function is established. The dynamics considered in \cite{Altman2007} considers marginal costs. In \cite{TouatiC09,Taouti09}, the replicator dynamics for particular allocation games are proved to converge to a pure Nash equilibrium by modifying game costs in order to obtain Lyapunov functions.

\section{Mean-Field Approximation For Generic Stochastic Algorithms}

Recall that we are interested in discussing the evolution of $Q(t)$, where $Q(t) = (q_1(t), ..., q_n(t)) \in K$ denotes the state of the player team at instant $t$ in the stochastic algorithm.

Clearly, $Q(t)$ is an homogeneous Markov chain.   
Define $\Delta Q(t)$ as $\Delta Q(t) =Q(t+1)-Q(t),$
and $\Delta q_{i}(t)$ as $q_i(t+1)-q_i(t).$ We can write 
\begin{equation} \label{eq:deuxd}
\Ec{\Delta q_i(t)}{Q(t)} = b p_i \Ec{F_i^b(r_i(t),s_i(t),q_i(t))}{Q(t)}, %\tilde{F}_i^b(Q(t)),
\end{equation}
%for function $\tilde{F}_i^b(Q,b)= \Ec{F_i^b(r_i(t),s_i(t),q_i(t))}{Q(t)}$, 
with $G_i(Q)=\lim_{b \to 0} \Ec{F_i^b(r_i(t),s_i(t),q_i(t))}{Q(t)}$ assumed to be continuous under our hypotheses.

Convergence of the stochastic algorithms towards ordinary differential equations defining their mean-field limit approximation
can be formalized as follows: Consider the piecewise-linear interpolation $Q^b(.)$ of $Q(t)$ defined by
% \begin{equation}
% \label{Piecewise_constant_interpolation}
$Q^b(t) = Q({\lfloor t/b\rfloor}) + (t/b - \lfloor t/b \rfloor) (Q({\lfloor t/b+1 \rfloor}) -  Q({\lfloor t/b\rfloor})). $
%\end{equation}
Function $Q^b(.)$ belongs to the space of all functions from $\R$ into
$K$ which are right continuous and have left hand
limits (\emph{cad-lag functions}). Now consider the sequence \{$Q^b(.) : b > 0$\}. We are
interested in the limit $Q(.)$ of this sequence when $b \to 0$.
Recall that a family of random variable $(Y_t)_{t \in \R}$ weakly
converges% \COMMENTAPPENDIX{strongconvergence}
to a random variable $Y$, if $E[h(X_t)]$ converges to
$E[h(Y)]$ for each bounded and continuous function $h$% \PROOFAPPENDIX{propindiens}{Theorem}
. 
% This is
% equivalent to convergence in distributions.

\begin{theorem}%[\cite{DecentralizedLearning94}]
  \label{propindiens}
The sequence of interpolated processes \{$Q^b(.)$\} converges weakly, when $b \to 0$, to $Q(.)$, which is the (unique deterministic) solution of initial value problem %\eqref{eq:replicator}
% \label{ODE_2_X}
\begin{equation} \label{eqgene}
\frac{dq_{i}}{dt} =   p_i G_i(Q), \mbox{ } i=1,\cdots,n,
%\frac{dq_{i,\ell}}{dt}= -q_{i,\ell} ( h_{i,\ell} - \overline{h_i}(Q)).
\end{equation}
with $Q(0) = Q^b(0)$.
\end{theorem}

To prove the theorem, we will use the following theorem from \cite[theorem 11.2.3]{stroock1979mdp}. The following presentation is inspired by the presentation of it in \cite[Theorem 5.8, page 96]{comets2006csm}.

Suppose that for all integers $b>0$, we have an homogeneous Markov
chain $(Y_k^{(b)})$ in $\R^d$ with transition kernel $\pi^{(b)}(x,dy)$, meaning that the law of $Y_{k+1}^{(b)}$, conditioned on
$Y_0^{(b)},\cdots,Y_k^{(b)}$, depends only on $Y_{k}^{(b)}$ and is
given, for all Borelian $B$, by
$P(Y_{k+1}^{(b)} \in B|Y_k^{(b)})=\pi^{(b)}(Y_k^{(b)},B),$
almost surely.

Define for $x \in \R^d$, 
\begin{align*}
  d^{(b)}(x) & = \frac1b \int (y-x) \pi^{(b)}(x,dy), \\
a^{(b)}(x) &= \frac1b \int (y-x)  (y-x)^* \pi^{(b)}(x,dy), \\
K^{(b)}(x) &=  \frac1b \int (y-x)^3 \pi^{(b)}(x,dy), \\
\Delta_\epsilon^{(b)}(x) &= \frac1b \pi^{(b)}(x,B(x,\epsilon)^c),
\end{align*}
where $B(x,\epsilon)^c$ denotes the complement of the ball with radius  $\epsilon$, centered at $x$. % In other words,
% $$
% d^{(b)}(x)=\frac1b \mathbb{E}_x [ (Y_1-x) ],
% $$
% and
% $$
% a^{(b)}(x) = \frac1b\mathbb{E}_x[(Y_1-x)(Y_1-x)^*]
% $$
% where $\mathbb{E}_x$ stands for ``expectation starting from $x$", that is,
% $$
% \mathbb{E}_x [ (Y_1-x) ]=\mathbb{E}[ (Y_1-x) | Y_0=x].
% $$

The coefficients $d^{(b)}$ and $a^{(b)}$ can be interpreted as the instantaneous drift and the variance (or matrix of covariance) of $X^{(b)}$.

Define 
 $$X^{(b)}(t) = Y^{(b)}_{\lfloor t/b \rfloor}  + (t/b - \lfloor t/b \rfloor) (Y^{(b)}_{\lfloor t/b+1 \rfloor} - Y^{(b)}_{\lfloor t/b\rfloor}).
  $$

\begin{theorem}[{\cite[theorem 11.2.3]{stroock1979mdp}, \cite[Theorem 5.8, page 96]{comets2006csm}}] \label{th:theo1}
Suppose that there exist some continuous functions $d,b$, such that for all $R<+\infty$,
\begin{align*}
  \lim_ {b\to 0} sup_{|x| \le R} |a^{(b)}(x)-a(x)| &= 0 \\
\lim_ {b\to 0} sup_{|x| \le R} |d^{(b)}(x)-d(x)| &= 0\\ 
\lim_{b\to 0} sup_{|x| \le R} \Delta_\epsilon^{(b)} &= 0, \forall \epsilon >0\\
\sup_{|x| \le R} K^{(b)}(x) &< \infty.
\end{align*}

With $\sigma$ a matrix such that $\sigma(x)\sigma^*(x)= a(x)$, $x \in \R^d$, we suppose that the stochastic differential equation 
\begin{equation}
\label{eq:diffusion}
dX(t) = d (X(t)) dt + \sigma(X(t)) dB(t), ~~~~ X(0)=x,
\end{equation}
has a unique weak solution for all $x$. This is in particular the case, if it admits a unique strong solution. 

Then for all sequences of initial conditions  $Y_0^{(b)} \to x$, the sequence of random processes $X^{(b)}$ \emph{weakly converges} to the diffusion given by Equation \eqref{eq:diffusion}. In other words, for all functions $F: \mathcal{C}(\R^+,\R) \to \R$ bounded and continuous, one has $$\lim_{b \to 0} E[F(X^{(b)})] = E[F(X)].$$
\end{theorem}

Theorem \ref{propindiens} follows from previous theorem. Consider $(Y_k^{(b)})$ to be 
$$Y_k^{(b)} =(Q(k))$$ with the corresponding $b$, which is indeed an homogeneous Markov chain. Let $\pi^{(b)}(Q,dy)$ be its transition kernel. %$\pi^{(b)}(Q,Q+ b F_i^b(r_i(t),s_i(t),q_i(t)) )

% For $x=Q% (\frac1p_1 q_1,\cdots,\frac1p_n q_n)
% $, 
We have
$$
\equations{
d^{(b)}_i(Q) & =&\frac 1b \int (y_i- q_i) \pi^{(b)}(x,dy) \\
&=& \frac1{b} \Ec{ \Delta q_i}{Q} \\
& =& \frac{p_i}{b}  b\tilde{F}_i^b(Q) \\
&\to& p_i G_i(Q) \mbox{ \hspace{1cm}  when $b \to 0$}\\
}$$
and
$$
\equations{
a^{(b)}_{i,j}(Q) &=& \frac1b \int (y_i- q_i)  (y_j- q_j)^* \pi^{(n)}(x,dy) \\
&=& \frac{b^2}{b} \Ec{p_i p_j \Delta q_i \Delta q_j}{Q} \\
&=& \mathcal{O}(b) \\
&\to & 0 \mbox{ when $b \to 0$}
}$$

In the same vein, clearly $K^{(b)}(Q) $ stay bounded, being in $\mathcal{O}(b^2)$. 

Now, from the fact that compact $K$ must be kept invariant by the dynamics,  $F_i^b(.)$ must have a compact support. This means that $\pi^{(b)}(Q,B(Q,\epsilon)^c)$ is $0$ for $b$ sufficiently small. Hence $\lim_{b\to 0} sup_{|x| \le R} \Delta_\epsilon^{(b)} =0$, $\forall \epsilon >0$. 

Hence, we have all the hypotheses of previous theorem with $a(Q)=0$ and $$d(Q)=(p_1G_1(Q),\cdots,p_nG_n(Q)),$$ observing that the corresponding stochastic differential equation $dQ(t) = d (Q(t)) dt + \sigma(Q(t)) dB(t)$ turns out to be an ordinary differential equation, whose solution is unique by (classical) Cauchy Lipschtiz theorem.

\section{General Games and Replicator-Like Dynamics}

From now on, we restrict to (possibly perturbed) replicator-like dynamics, as defined in page \pageref{def:replicatorlike}.

For replicator-like dynamics set $\alpha=1$ in what follows.

% For any such dynamic\DETAILSAPPENDIX{derivation}, 
For replicator-like dynamics  and perturbed replicator-like dynamics, the one-step dynamics of the stochastic algorithm can be
rewritten componentwise:

%\begin{equation} \label{composanteparcomposante}
$$
\Delta q_{i,\ell} (t) = q_{i,\ell}(t+1) - q_{i,\ell}(t) = \alpha
\left\{
\begin{array}{lllll}
0 & & + \mathcal{O}(b) &\mbox{ if } & i\neq i(t) \\
-% \alpha
b\gamma(r_i(t)) 
% (1-\bbb{r_i(t)}) 
q_{i,\ell} (t) && + \mathcal{O}(b) &\mbox{ if } &i=i(t)  \mbox{ and } s_i(t) \neq l\\
-% \alpha
b\gamma(r_i(t)) 
%(1-\bbb{r_i(t)}) 
q_{i,\ell} (t) + % \alpha
b(
\bb{ r_i(t)} ) && + \mathcal{O}(b) &\mbox{ if }  & i=i(t)   \mbox{ and } s_i(t) = l,\\
\end{array}
\right.
$$%\end{equation}
and we have
$$
\equations{
G_i(Q)  &=& \lim_{b \to 0} \frac1{b p_i} \Ec{\Delta q_{i,\ell}(t)}{Q(t)} \\
% \Ec{\Delta q_{i,\ell}(t)}{Q(t)} 
&=& \lim_{b \to 0}  \frac{1}b \sum_{j} q_{i,j}(t) \Ec{\Delta q_{i,\ell}(t)}{Q(t),s_i(t)=j,i(t)=i} \\
&=& + \alpha \sum_{j} q_{i,j}(t) q_{i,\ell}(t) \Ec{\gamma(r_i(t))}{Q(t),s_i(t)=\ell,i(t)=i}) \\
&&  - \alpha \sum_{j} q_{i,j}(t)  (q_{i,\ell}(t) \Ec{\gamma(r_i(t))}{Q(t),s_i(t)=j,i(t)=i}) \\
&=& q_{i,\ell} (\Ec{\gamma(r_i(t))}{Q(t),s_i(t)=\ell,i(t)=i}- \Ec{\gamma(r_i(t))}{Q(t),i(t)=i}).\\
% & & + 
%  \alpha q_{i,\ell}(t) 
%-     p_i\frac{b}{M}  q_{i,\ell}(t)\Ec{r_i(t)}{Q(t),s_i(t)=\ell,i(t)=i} \\
% % & = & p_i \sum_{j} q_{i,j}(t) (\frac{b}{M} q_{i,\ell} \Ec{r_i(t)}{Q(t),s_i(t)=j}) \\
% % & & - p_i q_{i,\ell}  \frac{b}{M} \Ec{r_i(t)}{Q(t),s_i(t)=\ell} -q_{i,\ell}\\
% % & = &  -\frac{b}{M} p_i q_{i,\ell} (\Ec{r_i(t)}{Q(t),s_i(t)=\ell} -
% % \sum_{j} q_{i,j}(t) )  \Ec{r_i(t)}{Q(t),s_i(t)=j})\\
% & = &  \frac{\alpha}{M} q_{i,\ell}(t) (\sum_{j} q_{i,j}(t) \Ec{r_i(t)}{Q(t),s_i(t)=j})  - \Ec{r_i(t)}{Q(t),s_i(t)=\ell,i(t)=i}) \\
}$$
that is to say, if we introduce $\uu_i(Q)=\Ec{-\frac1\alpha \gamma(r_i(Q))}{Q}$ for all $Q$, then
Equation \eqref{eqgene} leads to dynamics, 
%$$\frac{dq_{i,\ell}}{dt} = -p_i q_{i,\ell} ( \uu_i(e_{\ell},Q_{-i}) - \uu_i(q_i,Q_{-i})).$$
%
% \begin{equation} \label{eq:un}
% G_i(Q) = -\frac{1}{M} q_{i,\ell}(t) (\uu_i(e_{\ell},Q_{-i}) - \uu_i(q_i,Q_{-i})),
% %\Ec{\Delta q_{i,\ell}(t)}{Q(t)} = -\frac{b}{M} p_i q_{i,\ell}(t) (\uu_i(e_{\ell},Q_{-i}) - \uu_i(q_i,Q_{-i})).
% \end{equation}
% i.e.
by Theorem \ref{propindiens}.

% Equation \eqref{eqgene} leads to the following ordinary differential equation
This ordinary differential equation turns out to be (a rescaling of) (multipopulation) classical replicator dynamic
\begin{equation} \label{eq:replicatorres}
\frac{dq_{i,\ell}}{dt} = -p_i q_{i,\ell} ( \uu_i(e_{\ell},Q_{-i}) - \uu_i(q_i,Q_{-i})),
\end{equation}
whose limit points are related to Nash equilibria (through so-called Folk's theorems of evolutionary game theory \cite{Evolutionary1})% \PROOFAPPENDIX{th:deux}{Theorem}
.

Here, $\uu_i(Q)$ is taken as $\uu_i(Q)=\Ec{\gamma(r_i(Q))}{Q}$ for replicator-like dynamics, and $\uu_i(Q)=\Ec{\frac1\alpha \gamma(r_i(Q))}{Q}$ for perturbed replicator-like dynamics.  The game whose costs are defined by  $\uu_i$ is clearly isomorphic to the original game. Notice that when $\gamma$ is affine, this is just introducing a(n other) rescaling in \eqref{eq:replicatorres}.

Using properties of dynamics \eqref{eq:replicatorres}, we get:

% In any case, the game whose costs are given by $\uu_i(Q)$ is clearly isomorphic to the initial game, and hence what we get is that the dynamics behave, through an isomorphism, like dynamics \eqref{eq:replicatorres}.

\begin{theorem}\label{th:deux}
For general games, for any replicator-like or perturbed replicator-like dynamic, the sequence of interpolated processes \{$Q^b(.)$\} converges weakly, as $b \to 0$, to the unique deterministic solution of \eqref{eq:replicatorres} with $Q(0)=Q^b(0)$.
\emph{If} the mean-field approximation dynamic \eqref{eq:replicatorres}  converges, its stable limit points correspond to Nash equilibria of the game.
\end{theorem}

More precisely, we  have:

\begin{proposition} \label{prop:deux}
The following are true for the solutions  of Equation \eqref{eq:replicatorres}:
(i) All Nash equilibria are stationary points.
%\item All strict Nash equilibria are asymptotically stable.
(ii) All stable stationary points are Nash equilibria.  
(iii) However, (unstable) stationary points can include some Non-Nash equilibria.
\end{proposition}

The following are well-known (and obtained by just playing with definitions).

\begin{lemma}
A strategy profile $Q$ is a Nash Equilibrium iff $\uu_i(q_i,Q_{-i})  \le \uu_i(e_{\ell},Q_{-i})  $ for all $1 \le i \le n$, $1 \le \ell \le m_i$.
\end{lemma}

\begin{corollary} \label{coro:nash}
In a Nash Equilibrium, we have $\uu_i(q_i,Q_{-i})  = \uu_i(e_{\ell},Q_{-i})  $ for all $1 \le i \le n$, $1 \le \ell \le m_i$ with $q_{i,\ell}>0$.
\end{corollary}

Proposition \ref{prop:deux} is then an instance of the so-called folk-theorems of Evolutionary Game Theory \cite{Evolutionary1}. For completeness, the proof goes as follows:
From Corollary \ref{coro:nash}, clearly any Nash equilibria must also vanish the right-hand side of Equation \eqref{eq:replicatorres}.

A non-Nash equilibrium $Q$ is not stable: Indeed, if $Q$ is not a Nash equilibrium, this means that for some $i$, and some $\ell$ we have $\uu_i(q_i,Q_{-i})  > \uu_i(e_{\ell},Q_{-i})$. By bilinearity and continuity of $\uu_i$, function $\uu_i(q_i-e_{\ell},Q_{-i})$ must be strictly positive (say greater than $\epsilon$) on some neighborhood of $Q$. On this neighborhood, $\frac{dq_{i,\ell}}{dt}$ is greater than $p_i % \frac{\alpha}{M}
q_{i,\ell} \epsilon$, and hence the point is left exponentially faster  (faster than exponential $q_{i,\ell}(0)\exp(p_i % \frac{\alpha}{M}
\epsilon t)$).

 In a corner of $K$, we have for all $i$, $q_{i}=e_{\ell}$ for some $\ell$. Then clearly $q_{i,\ell'}=0$ for index $\ell'\neq \ell$, and $\uu_i(e_{\ell},Q_{-i}) - \uu_i(q_i,Q_{-i})=0$ for index $\ell'=\ell'$. Hence, the right-hand side of Equation \eqref{eq:replicatorres} is always null, and hence any corner is a stationary point.

More generally any state $Q$ in which all strategies in its support perform equally well, is clearly a stationary point from the definition of the dynamic.
%\end{proof}

% More precisely\COMMENTAPPENDIX{prop:deux}, the following are true for solutions  of \eqref{eq:replicatorres}:
% (i) All Nash equilibria are stationary points.
% %\item All strict Nash equilibria are asymptotically stable.
% (ii) All stable stationary points are Nash equilibria.  
% (iii) However, (unstable) stationary points can include some non-Nash equilibria.

Actually, all corners of simplex $K$ are stationary points, as well as, from the form of \eqref{eq:replicatorres}, more generally any state $Q$ in which all strategies in its support perform equally well.  Such a state $Q$ is not a Nash equilibrium as soon as there is an not used strategy (i.e. outside of the support) that performs better. 

Unstable limit stationary points may exist for the mean-field approximation \eqref{eq:replicatorres}% , and can not be avoided when talking about it
: Consider for example a dynamics that leave on some face of $K$ where some well-performing strategy is never used. %It may converge to some unstable stationary point of the dynamics.
%  Formally, this is easy to build examples whereas a whole set of points lead to dynamics that converge towards some unstable stationary point for dynamics \eqref{eq:replicatorres}: 
To avoid ``bad'' (non-Nash equilibrium, hence unstable) stationary points, following the idea of penalty functions for interior point methods, one can use as in  Appendix A.3 of \cite{sandholm2001pgc} some 
patches on the dynamics that would guarantee Non-complacency% \COMMENTAPPENDIX{modelpb}
. \emph{Non-Complacency (NC)} is the following property: $G(Q)=0$ implies that $Q$ is a Nash equilibrium \eqref{eq:replicatorres} (i.e. stationarity implies Nash).

This can be thought as the price to pay for purely deterministic models\footnote{And perhaps somehow as artifacts of modeling.}, and actually, when dealing with stochastic dynamics, all this can be avoided by taking profit of the unstability of non-Nash stationary points: this is the idea behind the randomized replicator dynamics already defined.
This guaranteed unstable points to be left almost-surely by the associated stochastic algorithm: technically, this ensures ergodicity of the underlying Markov Chain. Notice that a purely deterministic replicator-like dynamics where $\mathcal{O}(b)=0$ is not: an unstable stationary point, like a corner of $K$ is invariant for ever, and the underlying Markov is hence not  irreducible.

For general games, we get that the limit for $b \to 0$ is some ordinary differential equation whose stable limit points, when $t \to \infty$, \emph{IF} there exist, can only be Nash equilibria. Hence, \textit{IF} there is convergence of the ordinary differential equation, then one expects the previous stochastic algorithms to learn equilibria.% But for general games, there is no reason that there is always convergence.% , and the true problem, in addition to the double limit $b \to 0, t \to \infty$ is then to prove convergence of the ordinary differential equation.

Observe, that roughly speaking, for non-degenerated games, learning interior (hence mixed) Nash equilibria by such method is often problematic (and hence practically only pure Nash equilibrium may be learned) since the following is known:
\begin{proposition}[{\cite{amann1985plv,hofbauer1988tea},\cite[page 218]{LivreWeibull} }]
If a closed set $X \subset K$ belongs to the relative interior of some face of $K$, then $X$ is not asymptotically stable by dynamics \eqref{eq:replicatorres}.
\end{proposition}

\section{ Lyapunov Games, Ordinal and Potential Games} 

Since general games  have no reason to converge, we propose now to restrict to games for which replicator equation dynamic  or more generally general dynamics \eqref{eqgene} is provably convergent. 
As this practically often relies on some Lyapunov function argument, we propose the following terminology.

\begin{definition}[Lyapunov Game] 
We say that a game has a \emph{Lyapunov function} (with respect to a particular dynamic \eqref{eqgene} over $K$), or that the game is \emph{Lyapunov}, if there exists some non-negative $\mC^1$ function $F: K \to \R$ such that for all $i, \ell$ and $Q$, whenever $G(Q) \neq 0$, 
\begin{equation} \label{eq:lyapunov}
\sum_{i,\ell} p_i \frac{\partial F}{\partial q_{i,\ell}}(Q) G_{i,\ell}(Q)<0.
\end{equation}
\end{definition}

Lyapunov games include ordinal potential (and hence (exact) potential) games% \PROOFAPPENDIX{th:ordinal}{Theorem}
: we will say that a Lyapunov function $F: K \to \R$ is multiaffine, if it is defined as as polynomial in all its variables, it is of degree $1$ in each variable, and none of its monomials are of the form $q_{i,\ell}q_{i,\ell'}$.

\begin{theorem} \label{th:ordinal} An ordinal potential game is a Lyapunov game with respect to dynamics \eqref{eq:replicatorres}. Furthermore, its has some multiaffine  Lyapunov function.
\end{theorem} 

\begin{proof}
Consider $F(Q)=\Ec{\phi(Q)}{\mbox{ players play pure strategies according to probability distribution $Q$}}$ where $\phi$ is the potential of the ordinal potential game.  By linearity of expectation, $F(Q)$ is clearly multiaffine.

Now, by linearity of expectation, we have that  $F(q_i,Q_{-i})= \sum_{\ell} q_{i,\ell} F(e_{\ell},Q_{-i})$, and hence
$\frac{\partial F}{\partial q_{i,\ell}}(Q)= F(e_{\ell},Q_{-i})$. Now, for dynamics \eqref{eq:replicatorres}, left-hand side of \eqref{eq:lyapunov} rewrites to 
$$
%\begin{equation}
\equations{
\sum_{i,\ell} p_i \frac{\partial F}{\partial q_{i,\ell}}(Q) G_{i,\ell}(Q) &=&   - % \frac{\alpha}{M}
\sum_{i} p_i\sum_{\ell}  F(e_{\ell},Q_{-i}) q_{i,\ell} ( \uu_i(e_{\ell},Q_{-i}) - \uu_i(q_i,Q_{-i}))\\
& =& - % \frac{\alpha}{M}
\sum_{i} p_i \sum_{\ell} \sum_{\ell'} q_{i,\ell} q_{i,\ell'} F(e_{\ell},Q_{-i}) ( \uu_i(e_{\ell},Q_{-i}) - \uu_i(e_{\ell'},Q_{-i})) \\
&=& - \frac12 % \frac{\alpha}{M}
\sum_{i} p_i\sum_{\ell<\ell'} q_{i,\ell} q_{i,\ell'} (F(e_{\ell},Q_{-i})-F(e_{\ell'},Q_{-i})) ( \uu_i(e_{\ell},Q_{-i}) - \uu_i(e_{\ell'},Q_{-i})) \\
}
$$
%\end{equation}
Since the game is ordinal, $(F(e_{\ell},Q_{-i})-F(e_{\ell'},Q_{-i})) ( \uu_i(e_{\ell},Q_{-i}) - \uu_i(e_{\ell'},Q_{-i}))$ is always non-negative, by definition, and hence $F$ is a Lyapunov function.
\end{proof}

More precisely, if $\phi$ is the potential of the ordinal potential game, then one can take its expectation  $F(Q)=\E{\phi(Q)}=\Ec{\phi(Q)}{\mbox{ players play pure strategies according to } Q}$%, denoted by $\E{\phi}$.
as a Lyapunov function with respect to dynamics \eqref{eq:replicatorres}.

The following class of games have been introduced%\PROOFAPPENDIX{prop:lyapunovpot}{Proposition}  \cite{sandholm2001pgc,DecentralizedLearning94}.

\begin{definition}[Potential Game \cite{sandholm2001pgc}]
A game is called a \emph{continuous potential game} if there exists a $\mC^1$ function $F: K \to \R$ such that
for all $i,\ell$ and $Q$, 
\begin{equation} \label{eq:potentiel}
\frac{\partial F}{\partial q_{i,\ell}}(Q) = \uu_i(e_\ell,Q).
\end{equation}
%for all $i, \ell$ and $Q$.
\end{definition}

\begin{proposition} \label{prop:lyapunovpot}
A continuous potential game is a Lyapunov game with respect to dynamics \eqref{eq:replicatorres}. Furthermore, its has some multiaffine  Lyapunov function.
\end{proposition}

\begin{proof}

By definition, $F$ has a multiaffine Lyapunov function: this is clear as all its partial derivative are known, given by $\frac{\partial F}{\partial q_{i,\ell}}(Q) = \u_i(e_\ell,Q).$

Now, in this case, for dynamics \eqref{eq:replicatorres}, left-hand side of \eqref{eq:lyapunov} rewrites to 
$$
%\begin{equation}
\equations{
\sum_{i,\ell} p_i \frac{\partial F}{\partial q_{i,\ell}}(Q) G_{i,\ell}(Q) &=&   - % \frac{\alpha}{M}
\sum_{i} p_i \sum_{\ell}  \uu_i(e_{\ell},Q_{-i}) q_{i,\ell} ( \uu_i(e_{\ell},Q_{-i}) - \uu_i(q_i,Q_{-i}))\\
& =& - % \frac{\alpha}{M}
\sum_{i} p_i\sum_{\ell} \sum_{\ell'} q_{i,\ell} q_{i,\ell'} \uu_i(e_{\ell},Q_{-i}) ( \uu_i(e_{\ell},Q_{-i}) - \uu_i(e_{\ell'},Q_{-i})) \\
&=& - \frac12 % \frac{\alpha}{M}
\sum_{i} p_i \sum_{\ell<\ell'} q_{i,\ell} q_{i,\ell'} (\uu_i(e_{\ell},Q_{-i})-\uu_i(e_{\ell'},Q_{-i}))^2 % ( \uu_i(e_{\ell},Q_{-i}) - \uu_i(e_{\ell'}i,Q_{-i}))
\\
}
$$
hence is positive on non-stationary points.
\end{proof}

Recall that exact potential games have been defined page \pageref{def:potential}, following \cite{MondererShapley1996}, in terms of pure strategies. Notions turn out to be equivalent%\PROOFAPPENDIX{prop:potentialpotential}{Proposition}
 when $F$ is assumed at least $\mC^2$.

\begin{proposition} \label{prop:potentialpotential}
An (exact) potential game 
of potential $\phi$ leads to a continuous potential game with $F(Q)=\E{\phi(Q)}$, and conversely, the restriction of $F$ of class $\mathcal{C}^2$  to pure strategies of a potential in the sense of above definition leads to an (exact) potential.
\end{proposition}

\begin{proof}

% A specific case of Lyapunov games is obtained trough potential games. The notion of potential game has been introduced by Monderer and Shapley in \cite{MondererShapley1996} for games with a finite number of players, and then extended by Sandholm \cite{sandholm2001pgc} to games with continuous sets of players.

% \begin{definition}[Potential Game \cite{sandholm2001pgc}]
% A game is called a potential game if there exists a $\mC^1$ function $F: K \to \R$ such that
% for all $i,\ell$
% \begin{equation} \label{eq:potentiel}
% \frac{\partial F}{\partial q_{i,\ell}}(Q) = \uu_i(e_\ell,Q),
% \end{equation}
% for all $i, \ell$ and $Q$.
% \end{definition}
In other words, a game is a continuous potential game if there exists some $\mC^1$ function whose gradient $\nabla f$ equals the cost vector $H=(\uu_i(e_l,Q))_{i,l}$. Function $F$, which is unique up to an additive constant, is called the potential function of the game.

When $F$ is $\mC^2$, condition \eqref{eq:potentiel} is equivalent to \emph{externality symmetry} \cite{sandholm2001pgc,MondererShapley1996}: 
\begin{equation} \label{eq:ES}
\frac{\partial \uu_{i}(e_\ell,Q)}{\partial q_{j,\ell'}} = \frac{\partial \uu_{j}(e_\ell',Q)}{\partial q_{i,\ell}},
\end{equation}
for all $i,j,\ell,\ell'$.
In that case, by a well-known result (characterization of exact forms), if we fix any $z \in K$, $F$ is given by 
\begin{equation} \label{eq:integrale}
F(Q) = \sum_{i=1}^{n} \sum_{\ell=1}^{m_i}\int_0^1 \uu_i(e_\ell,x(t)) x'_i(t) dt,
\end{equation}
where $x:[0,1] \to K$ is any piecewise continuous differentiable path in $K$ that connects $z$ to $Q$ (i.e. $x(0)=z$, $x(1)=Q$). 

In particular it must be independent of the used path. Considering paths from pure strategies to pure strategies, the second part of the proposition follows, from characterizations of (exact) potential games  in \cite{MondererShapley1996}.
The first part of the proposition is easy to establish, in the same vein as we established $\frac{\partial F}{\partial q_{i,\ell}}(Q)= F(e_{\ell},Q_{-i})$  in the proof of Theorem  \ref{th:ordinal} above. 
\end{proof}

A Lyapunov game can have some non-multiaffine potential function, hence not all Lyapunov games with respect to dynamics \eqref{eq:replicatorres} are ordinal games. We believe Lyapunov game with respect to dynamics \eqref{eq:replicatorres} with a multiaffine potential function to differ from ordinal games.

The interest of Lyapunov functions is that they provide convergence.  Recall that the $\omega(Q_0)$ limit set of a point $Q_0$ is the set of accumulation points of the trajectories that start from $Q_0$: considering a  trajectory starting from $Q_0$, this is the set of $Q^*$ with $Q^*=\lim_{n \to \infty} Q(t_n)$, for some increasing sequence $(t_n)_{n \ge 0} \in \R$% \PROOFAPPENDIX{prop:lyapunov}{Proposition}
.

\begin{proposition}\label{prop:lyapunov}
In any Lyapunov game with respect to any dynamic \eqref{eqgene} over $K$,  the solutions of mean-field approximation \eqref{eqgene} % converge to some stationary point of the dynamics.
have   their limit set $\omega(Q)$ non-empty, compact, connected, and consisting entirely of stationary points of the dynamic. On this limit sets, $F$ is constant.
\end{proposition}
\begin{proof}
This is made of well-known fact, and is for example present for example as Lemma A.1 of \cite{sandholm2001pgc}.

For self-contentedness, here is mainly a slight adaptation of the proof of Lyapunov Stability theorem \cite[page 194]{HSD03}.

$F(Q(t))$ must be monotone along trajectories, since Equation \eqref{eq:lyapunov} guarantees $\frac{dF(Q(t))}{dt}<0$. Let $Q(t)$ be some solution of ordinary differential equation \eqref{eqgene} with $Q(0)=x$. Let $Q_0 \in \omega(x)$: that is to say $Q(t_n) \to Q_0$ for some sequence $t_n \to \infty$. We claim that $Q_0$ must be some stationary point of the dynamics, that is to say, $G(Q_0)=0$. To see this, observe that $F(Q(t)) > F(Q_0)$ since $F(Q(t))$ decreases and $F(Q(t))$ converges to $F(Q_0)$ by continuity of $F$. 

Suppose that $G(Q_0) \neq 0$. Let $Z(t)$ be the solution of the ordinary differential equation starting from $Q_0$. For any $s>0$, we have $F(Z(s)) <F(Q_0)$. Hence, for any solution $Y(s)$ starting sufficiently near $Z_0$ we have $F(Y(s)) < F(Q_0)$. Setting $Y(0)=Q(t_n)$ for sufficiently large $n$ yields the contradiction $F(Q(t_n+s)) < F(Q_0)$. Therefore, $G(Q_0)=0$. 

This proves that any limit set must be non-empty and consisting entirely of stationary point of the dynamics.

By continuity of $F$, $F(Q_0)= \lim_{n \to \infty} F(Q(t_n))$ for any limit point $Q_0$. Now this must be equal to $\inf_t(F(Q(t)))$ and hence independent of $Q_0$.

The subset $\omega(x)$ of limit points $Q_0$, being equal to $\cap_{t} Closure(F(s \ge t))$, hence a decreasing intersection of compact connected sets must be compact and connected.

%  Observing that $K$ is compact, from Lyapunov Stability theorem  \cite[page 194]{HSD03} 
% all trajectories will converge to points in the set $K'=\{Q^* \in K:\frac{dF(Q^*)}{dt}=0\}$. By hypothesis, $dF(Q^*)$ can be $0$ only if $G(Q)=0$, i.e. $K'$ is made of stationary points of the dynamics. Now, a limit point must be stable.
\end{proof}

Observing that all previous classes are Lyapunov games with respect to dynamics \eqref{eq:replicatorres}, this gives the full interest of this corollary.

\begin{corollary}
In a Lyapunov game with respect to general dynamics \eqref{eq:replicatorres}, whatever the initial condition is, the solutions of mean-field approximation \eqref{eqgene} will converge. The stable limit points are Nash equilibria.
\end{corollary}

If mean-field approximation \eqref{eqgene} has the (NC) property, then this guarantees that limit points are Nash equilibria. Otherwise, unstable limit stationary may exist for the mean-field approximation. % Such points will be left almost-surely if by the underlying stochastic dynamics if it is perturbed: see coming discussions.

\section{Replicator-Like Dynamics for Multiaffine Lyapunov Games}

Fortunately, this is possible to go further, observing that many of the previous classes (ordinal, (exact) potential, continuous potential, load balancing games, congestion games, task allocation games)  turn out by previous discussion to have a multiaffine Lyapunov function.

When this holds, this is indeed possible to talk directly about the stochastic algorithms, avoiding % these problems and the
passage through ordinary differential equation \eqref{eqgene}, and the double limit $b \to 0$, $t \to \infty$. % in some specific important cases.
The key observation is the following (the proof mainly relies on the fact that second order terms are null for multiaffine functions)% \PROOFAPPENDIX{prop:keylemma}{Lemma}
.

% Following an idea in appendix of \cite{sandholm2001pgc}, we can even talk about patched dynamics \eqref{eq:cpenaly}.

% \begin{proposition} 
% Any ordinal potential game of potential is a Lyapunov game for a dynamics of the form \eqref{eq:cpenaly}. Furthermore, its has some multiaffine  Lyapunov function.
% \end{proposition}

\begin{lemma}\label{prop:keylemma}
 When $F$ is a multiaffine Lyapunov function, 
\begin{equation}  \label{eq:demi}
\Ec{\Delta F(Q(t+1))}{Q(t)} =\sum_{i=1}^n \sum_{\ell=1}^{m_i} \frac{\partial F}{\partial q_{i,\ell}} (Q(t)) \Ec{\Delta q_{i,\ell}}{Q(t)},
\end{equation}
where $\Delta F(t)= F(Q(t+1))-F(Q(t))$.
\end{lemma}
\begin{proof}

Let us denote 
$R(Q,\Delta )= F(Q+\Delta)-F(Q)-\sum_{i=1}^n \sum_{\ell=1}^{m_i} \frac{\partial F}{\partial q_{i,\ell}}(Q) \Delta_{i,\ell}$
when $\Delta$ is a vector, so that by definition taking $\Delta=\Delta Q(t)$,  we have
$\Delta F(t)= F(Q(t+1))-F(Q(t))= \sum_{i=1}^n \sum_{\ell=1}^{m_i} \frac{\partial F}{\partial q_{i,\ell}} (Q) \Delta q_{i,\ell} + R(Q,\Delta Q(t)).$

We then have
$$
\Ec{\Delta F(t)}{Q(t)} = \sum_{i=1}^n \sum_{\ell=1}^{m_i} \frac{\partial F}{\partial q_{i,\ell}} (Q) \Ec{\Delta q_{i,\ell}}{Q(t)}  +  \Ec{R(Q,\Delta Q(t))}{Q(t)}.
$$

It only remains to prove that $\Ec{R(Q,\Delta Q(t))}{Q(t)}=0$ when $F$ is multiaffine.

A multiaffine function $F$ is particular polynomial function, of degree $1$ in each variable. By definition, $R(Q,\Delta Q)$ is hence also a polynomial function, of degree $1$ in each variable $\Delta Q_{i,\ell}$. By construction, it has no-constant term, and no monomial of the form $\beta_{i,\ell} \Delta Q_{i,\ell}$. Hence, all its monomials are of the form $\beta_{i,\ell,j,\ell'}\Delta Q_{i,\ell} (t) \Delta Q_{j,\ell'} (t)$, with $(i,\ell) \neq (j,\ell')$.

By definition of multiaffine function used in this paper,  there can not be terms $\Delta Q_{i,\ell} (t) \Delta Q_{j,\ell'} (t)$ with $i=j$ among these monomials.

Observe that $\Delta Q_{i,\ell}(t) \Delta Q_{j,\ell'}(t)=0$ for $i \neq j$:  indeed, at any time $t$, at most one player moves in the considered class of algorithms: in other words, we use the fact that considered algorithms are elementary stepwise.
\end{proof}

When considering a Lyapunov game with respect to replicator-like dynamics, using Equation \eqref{eq:deuxd} and the fact that $G_i(Q)=\lim_{b \to 0} \tilde{F}_i^b(Q)$ % , and the fact that the partial derivative of $F$ are bounded over compact $K$, 
the right hand side of Equation \eqref{eq:demi} is 
\begin{equation} \label{eq:ideale}
b \sum_{i=1}^n \sum_{\ell=1}^{m_i} p_i \frac{\partial F}{\partial q_{i,\ell}} (Q) G_{i,\ell}(Q)  +\mathcal{O}(b^2),
\end{equation}
and hence expected to be negative by Equation \eqref{eq:lyapunov} when $G(Q)\neq 0$ and $b$ is sufficiently small. 

In other words, when $b$ is small, $(F(Q(t))_{t}$ will be a super-martingale until reaching a point where \eqref{eq:ideale} is close to $0$.

More precisely, for a replicator-like dynamics, Equation \eqref{eq:ideale} rewrites to% \SEEAPPENDIX{prop:lyapunovpot}

$$ 
- b \frac14 \sum_{i} p_i \sum_{\ell\neq \ell'} q_{i,\ell} q_{i,\ell'} (\uu_i(e_{\ell},Q_{-i})-\uu_i(e_{\ell'},Q_{-i}))^2 +\mathcal{O}(b^2).
$$

\OPTIMISTE{
As expected, on corners of $K$, this is expected to be close to $0$, and hence not (neccesarily) a super-martingale. 

For the perturbed replicator-like dynamics, taking the perturbation $\mathcal{O}(b)$ in page \pageref{per:urgence} to be $0$, Equation \eqref{eq:ideale} rewrites to
$$ 
- b \alpha \frac14 \sum_{i} p_i \sum_{\ell\neq \ell'} q_{i,\ell} q_{i,\ell'} (\uu_i(e_{\ell},Q_{-i})-\uu_i(e_{\ell'},Q_{-i}))^2
+ b^2 (1-\alpha) \sum_{i=1}^n \sum_{\ell=1}^{m_i} \frac{\partial F}{\partial q_{i,\ell}} (Q(t)) (\frac{1}{m_i} -q_{i,\ell}).
$$
which can be written
$$
- b \alpha \frac14 \sum_{i} p_i \sum_{\ell\neq \ell'} q_{i,\ell} q_{i,\ell'} (\uu_i(e_{\ell},Q_{-i})-\uu_i(e_{\ell'},Q_{-i}))^2
+ \mathcal{O}(b^2).$$
}

% % \EXPERIMENTAL{
% % For pertubed replicator-like dynamics, taking the perturbation $\mathcal{O}(b)$ to be $0$, we have

% % $$\Ec{\Delta q_{i,\ell}}{Q(t)} = - \alpha b  p_i q_{i,\ell} (u_i(e_{\ell},Q_{-i})-u_i(q_i,Q_{-i})) + b^2
% % (1-\alpha) (\frac{1}{m_i} -q_{i,\ell}).
% % $$
% % Let $\zeta=-  p_i (u_i(e_{\ell},Q_{-i})-u_i(q_i,Q_{-i}))$. This is of the form
% % $$ \alpha b q_{i,\ell} \zeta + b^2
% % (1-\alpha) (\frac{1}{m_i} -q_{i,\ell}).
% % $$

% % This is greater than $\Omega \alpha b \zeta $ as soon as $\zeta>1$, and $b$ is taken sufficiently small.
% % }

% \EXPERIMENTALDEUX{
% Notice that such a function is always positive, and bigger than
% $\min(b^2 (1-\alpha),- \alpha b  p_i q_{i,\ell} (u_i(e_{\ell},Q_{-i})-u_i(q_i,Q_{-i}))$ whenever
% $-  p_i q_{i,\ell} (u_i(e_{\ell},Q_{-i})-u_i(q_i,Q_{-i}))$ is positive. 

% In particular, if $\zeta=-  p_i q_{i,\ell} (u_i(e_{\ell},Q_{-i})-u_i(q_i,Q_{-i}))$ is greater than $1$, this is in $\omega(\zeta)$. 
% }

% \EXPERIMENTAL{

% The second term is of order $\omega(b^2)$. 

% In the first term, if there is some $q_{i,\ell} >0$ and some $q_{i,\ell'}>0$, its opposite is greater than
% $$q_{i,\ell} q_{i,\ell'} (\uu_i(e_{\ell},Q_{-i})-\uu_i(e_{\ell'},Q_{-i}))^2.$$

% In particular, if $q_{i,\ell} >\delta$, this is greater than
% $$q_{i,\ell'} \delta (\uu_i(e_{\ell},Q_{-i})-\uu_i(e_{\ell'},Q_{-i}))^2.$$
% }

% Recall discussion about unstable stationary points.  
When talking about stochastic perturbed dynamics, using this super-martingale argument, one gets the following stability result% \PROOFAPPENDIX{th:timetwo}{Proposition}:
: we write $L(\mu)$ for the subset of states $Q$ on which $F(Q) \le \mu$.

\begin{proposition}\label{th:timetwo}
Let  $\lambda>1$. Let $Q(t_0)$ be some state.  Consider $b$ enough small so that \eqref{eq:demi} is non-positive outside of $L(F(Q(0)))$. Then $Q(t)$ will be such that $Q(t) \in L(\lambda F(Q(t_0)))$  forever  after time $t \ge t_0$ with a probability greater than $1-\frac1\lambda$.
\end{proposition}

\begin{proof}
% Assume that $Q(0) \in I(\epsilon)$. 
Consider sequence $Z_n=\max_{t \le n} F(Q(t))$ and $\mathcal{F}_i$ the sigma-algebra generated by $(Q(j))_{j \le i}$, and apply Proposition \ref{prop:doob} for $\lambda'=\lambda \E{Z_0}$:  $$P[\forall n, F(Q(n)) \le \lambda F(Q(0)]=P[\sup_{n} Z_n \ge \lambda'] \le \frac{\E{Z_0}}{\lambda'} = \frac1\lambda.$$

% \subsection{Proof of Theorem \ref{th:almostsure}}
% \label{proof:th:almostsure}

% Let $F^*$ be the minimum value of $F$ over compact $K$. Without loss of generality (substracting $F^*$ to the potential if needed), assume that $F^*=0$. 

% A point realizing minimum value of $F$ must correspond to a Nash equilbrium. Indeed, since $F$ is a Lyapunov function, and $F^*$ is its minimum, points that realize this minimum are stable. It follows from the reasoning in the proof of Proposition \ref{prop:deux} that they correspond to Nash equilibria. 

% Let $\epsilon>0$. We show that the property holds with probability greater than $1-1/\lambda$ for all $\lambda$. Fix $\lambda$. 
% Let $V_\lambda$ be some neighborhood on which $F^*- F(Q) \le \epsilon/\lambda$. Since the dynamics is perturbed, the underlying Markov chain is ergodic, and hence almost-surely $V_\lambda$ will be visited. Once there, Proposition \ref{th:timetwo} says that $F(Q(t))$ will stay less than $\epsilon$ forever with probability greater than  $1-1/\lambda$.

% Notice that when $F$ is multiaffine, $F(Q) - F(Q^*)=\nabla F(Q^*).(Q-Q^*)$, where $\nabla F$ denotes gradient, and hence that this is linearly related to the distance between $Q$ and $Q'$.

% The result follows.

% % AVEC LES MAINS.
\end{proof}

If dynamic is perturbed, then the underlying Markov chain is ergodic. It follows that any neighborhood is visited with a positive probability: a dynamic will be said \emph{perturbed} if for all $Q \in K$, for any neighborhood $V$ with $Q$ in its closure, the probability that $Q(t+1) \in V$ when $Q(t)=Q$ is positive. % We say that a point $Q$ is $\epsilon$-close to some Nash equilibrium $Q^*$, if $|Q - Q^*| \le \epsilon$. 

Then if in some neighborhood of such a point we can apply previous proposition, one would get that almost surely, after some time, $Q(t)$ will be close to some Nash equilibria forever with high probability. The default of such an approach is clearly on the fact that it does not provide bounds on the time required to reach such a neighborhood.

Notice that for Lyapunov game with a multiaffine Lyapunov function $F$, with respect to Dynamic~\eqref{eq:replicatorres} (this include ordinal, and hence potential games from above discussion), the points $Q^*$ realizing the minimum value $F^*$ of $F$ over compact $K$ must correspond to Nash equilibria.

% We get\PROOFAPPENDIX{th:almostsure}{Theorem}: A dynamics 

% \begin{theorem} \label{th:almostsure}

% Consider a Lyapunov game with a multiaffine Lyapunov function $F$, with respect to \eqref{eq:replicatorres}. This includes ordinal, and hence potential games from above discussion.  The points $Q^*$ realizing the minimum value $F^*$ of $F$ over compact $K$ correspond to Nash equilibria. If on some 

% Fix $\epsilon>0$. Fix $\lambda>0$. Assume that one can find a neighborhood of the set of these points so that $F(Q) - F^* \le \epsilon/\lambda$ on this neighborhood.

% Taking $b$ sufficiently small so that  \eqref{eq:demi} is non-positive on this neighborhood. Then, using a perturbed replicator-like dynamics, whatever the initial condition is, for all $\lambda>0$, taking $b$ sufficiently small 

% so that in some neighborhood of the minimum of $F$ one has \eqref{eq:demi} non-positive, then 
% almost surely, after some time, $Q(t)$ will be $\epsilon$-close to some Nash equilibria forever.
% \end{theorem}

Fortunately, this is possible to get bounds on the expected time of convergence% \PROOFAPPENDIX{th:timeone}{Theorem}
: we write $L(\mu)$ for the subset of states $Q$ on which $F(Q) \le \mu$.

\begin{definition}[$\epsilon$-Nash equilibrium]
Let $\epsilon \ge 0$. 
A state $Q$ is some $\epsilon$-Nash equilibrium iff for all $1 \le i \le n, 1 \le \ell \le m_i$, we have 
$\uu_i(e_{\ell},Q_{-i}) \ge (1-\epsilon) \uu_i(q_i,Q_{-i}).$
\end{definition}

\OPTIONAL{
If one prefers, in an $\epsilon$-Nash equilibrium, no player can improve its situation by more than $\epsilon$ times its current cost by changing unilaterally its strategy. %A Nash equilibrium corresponds to $\epsilon=0$. 

In a 
non $\epsilon$-Nash equilibrium, we have some $i$ and $\ell$, with 
$\uu_i(e_{\ell},Q_{-i}) < (1-\epsilon) \uu_i(q_i,Q_{-i})$. This means, 
$\uu_i(q_i-e_{\ell},Q_{-i}) > \epsilon \uu_i(q_i,Q_{-i})$.
}

For the perturbed replicator-like dynamics, taking the perturbation $\mathcal{O}(b)$ to be $0$ in the definition of this dynamics,  we have $$\Ec{\Delta q_{i,\ell}}{Q(t)} = - \alpha b  p_i q_{i,\ell} (u_i(e_{\ell},Q_{-i})-u_i(q_i,Q_{-i})) + b^2
(1-\alpha) (\frac{1}{m_i} -q_{i,\ell}).
$$

Assume without loss of generality that all costs are greater than $1$. 
Let $\zeta=p_i u_i(q_i-e_{\ell},Q_{-i})$ and $\beta=1-\alpha$. Previous equation is of the form
$ b(\alpha  q_{i,\ell} \zeta + b
\beta (\frac{1}{m_i} -q_{i,\ell})),$
hence some strictly increasing function of $q_{i,\ell}$ as soon as $b < \frac\alpha\beta p_i \epsilon \uu_i(q_i,Q_{-i})$ and $\zeta> p_i\epsilon \uu_i(q_i,Q_{-i})$. In that case, its minimal value, obtained for $q_{i,\ell}=0$ is $\delta=\frac{b^2\beta}{m_i}$. 

So, as soon as $\zeta > p_i \epsilon \uu_i(q_i,Q_{-i})$, that is to say $u_i(q_i-e_{\ell},Q_{-i}) > \epsilon \uu_i(q_i,Q_{-i})$, we will have $\Ec{\Delta q_{i,\ell}}{Q(t)} \ge \delta$, that implies $\Ec{q_{i,\ell}(t+1)}{Q(t)} \ge \delta$.

This implies that the opposite of $\Ec{\Delta F(Q(t+1))}{Q(t)}$ will be greater than 
$$V= b \alpha \frac14 p_i \delta \sum_{\ell\neq \ell'} q_{i,\ell'} (\uu_i(e_{\ell},Q_{-i})-\uu_i(e_{\ell'},Q_{-i}))^2
+ \mathcal{O}(b^2).
% - b^2 \beta \sum_{i=1}^n \sum_{\ell=1}^{m_i} \frac{\partial F}{\partial q_{i,\ell}} (Q(t+1)) (\frac{1}{m_i} -q_{i,\ell})
$$

Taking $b < (1-\mu) \frac\alpha\beta p_i \epsilon \uu_i(q_i,Q_{-i})$ for any $\mu>0$ guarantees that the factor in $q_{i,\ell}$ in previous discussed expression is greater than $\mu \zeta \alpha$, and hence that its iterations growth exponentially fast near $0$.  Reasoning by sequences of $k$ steps, i.e. about 
%  If it holds during $k$ consecutive steps, then $\Ec{q_{i,\ell}(t+k)}{Q(t)}$ will growth exponentially faster that $k$ iterations of previous functions. It follows that any neighborhood where this hold must be left in a finite number of steps. 
% It follows that after a finite expected time,
the opposite of $\Ec{\Delta F(Q(t+k))}{Q(t)}$, 
% the term in $\mathcal{O}(b^2)$ can be neglicted above. 
will greater than a term of order
$$V= b \alpha \frac14 p_i  \sum_{\ell\neq \ell'} q_{i,\ell'} (\uu_i(e_{\ell},Q_{-i})-\uu_i(e_{\ell'},Q_{-i}))^2
% % - b^2 \beta \sum_{i=1}^n \sum_{\ell=1}^{m_i} \frac{\partial F}{\partial q_{i,\ell}} (Q(t+1)) (\frac{1}{m_i} -q_{i,\ell})
 $$
in a non-$\epsilon$-Nash equilibrium.

\begin{theorem} \label{th:timeone}
Consider a Lyapunov game with a multiaffine Lyapunov function $F$, with respect to \eqref{eq:replicatorres}. This includes ordinal, and hence potential games from above discussion. Taking $b=\mathcal{O}(\epsilon)$, 
whatever the initial state of the stochastic algorithm is, it will almost surely reach % $I(\epsilon) % \cup L(\mu)
% $.
some $\epsilon$-Nash equilibrium. 
Furthermore, it will do it in a random time whose expectation $T(\epsilon)$ satisfies 
$$T(\epsilon) \le \mathcal{O}(\frac{F(Q(0))}{\epsilon}).$$
\end{theorem}
\begin{proof}

 Consider $V^*=\min_{i,q_i} \sum_{\ell\neq \ell'} q_{i,\ell'} (\uu_i(e_{\ell},Q_{-i})-\uu_i(e_{\ell'},Q_{-i}))^2$.  Let $I(\epsilon)$ denote the states where the righthand side of Equation \eqref{eq:demi} is greater than $- b\alpha \frac14 \min_i p_iV^* 
\epsilon$.

If the initial state is already $\epsilon$-stable % or in $L(\mu)$,
 then there is nothing to prove.

Otherwise, this follows from the analysis before Theorem \ref{th:timeone}, and from proposition \ref{prop:cinq}, with $Z_i= F(Q(i))$, $\mathcal{F}_i$ the sigma-algebra generated by $(Q(j))_{j \le i}$, $C=\mu$, $K=I(\epsilon)$: indeed, whenever
$Q(t) \not\in I(\epsilon) \cup L(\mu)$, this implies $\tau > t$,  and we have $\Ec{\Delta F(t)}{Q(t)} = \Ec{Z_{t+1}-Z_{t}}{\mathcal{F}_t} \le -\epsilon \mathcal{O}(b)$. In all other cases, 
$\Ec{\tilde{Z}_{n+1} }{\mathcal{F}_n}= \tilde{Z}_n$ and hence all the hypotheses of Proposition \ref{prop:cinq} are satisfied.
\end{proof}

% Observe that, once in an $\epsilon$-Nash equilibrium at time $t_0$, i.e. \textbf{a local minimum of $F$}, it we will stay with high probability, since we have the following 

\CONJECTURAL{ We believe these bounds are tight for generic ordinal games.  The point is that in arbitrary ordinal games, there is no necessarily relation between the gain in utility and the gain in potential: only sign of variation must be preserved. 

Of course better bounds can be hoped for particular games, and in particular for congestion games. For generic congestion games, there is a strong relation between the potential and utilities of players. In congestion games, using notations from page \pageref{def:congestion}, the potential is given by $F(Q) = \E{ \sum_{r=1}^m \sum_{t=1}^{\lambda_r(t)} C_r(t)}$. One has in particular $F(Q) \le \E{\sum_{i=1}^n \uu_i(c_i,Q)},$ since $c_i (Q) = \sum_{r \in q_i} C_r (\lambda_r (Q)).$ 
%
% For potential games, we believe this is really possible to go further.

In particular, following \cite{ChienS07},  a congestion game is said to satisfy the \emph{$\alpha$-bounded jump condition} if its cost functions satisfy $C_r(t+1) \le \alpha C_r(t)$ for all $t \ge 1$. This ensures the following property % \COMMENTAPPENDIX{deltaratio}
for $\delta=\frac1{\alpha n}$ (see \cite{ChienS07}): 
%
% When this holds, the right hand  side of \eqref{eq:demi} is greater than $-\epsilon b + \mathcal{O}(b^2) \epsilon = -\epsilon \mathcal{O}(b)$, when $b$ is sufficiently small (less than some constant).
%
% \begin{definition}[$\delta$-Ratio-property] $\delta$-Ratio-property  is the following property: 
  whenever $Q$ is not an $\epsilon$-Nash equilibrium, then for at least a player $i$, the relative cost of adopting some pure strategy $\ell$ would induce a gain at least $\delta$ times the resulting gain in potential. % This guarantees logarithmic convergence of potential. 
  % : $\Delta\uu_i \ge \gamma \Delta F $, where $\Delta \uu_i= \uu_i(q_i,Q_{-i})- \uu_i(e_{\ell},Q_{-i}) $ denotes the gain in cost, and $\Delta F= F(q_i,Q_{-i}) - F(e_{\ell},Q_{-i})$ the gain in potential. 
% % \end{definition}

% \begin{proposition} Symmetric congestion games have the $$-Ratio property.
% \end{proposition}

% In a 
% non $\epsilon$-nash equilibrium, we then have some $i$ and $\ell$, with 
% $\uu_i(e_{\ell},Q_{-i}) < (1-\epsilon) \uu_i(q_i,Q_{-i})$. 

We believe perturbed replicator-like dynamics to converge very fast (hence in polynomially many steps)
on such games.}

\newpage
\appendix

\section{Results About Semi-Martingales}
\label{proof:th:timeone}

Let $\{Z_i, i \ge 0\}$ be a sequence of real non-negative random variables, such that $Z_i$ is measurable in the increasing family of sigma-algebra $\mathcal{F}_i$. 

\begin{proposition}[{proof similar to \cite[Theorem 2.1.1, page 17]{LivreFayolle}}] \label{prop:cinq}
Assume that $Z_0$ is constant. Denote by $\tau$ the $\mathcal{F}_{n}$-stopping time representing the epoch of the first entry into $[0,C]$ or in some measurable subset $K$, for $C>0$, i.e. $\tau(\omega)=\inf \{n \ge 1| Z_n(\omega) \le C \vee  Z_n(\omega) \in K\}$. Introduce the stopped sequence
$$\tilde{Z}_n = Z_{n \wedge \tau},$$ where
$$n \wedge \tau = \left\{ \begin{array}{ll}
n, &\mbox{ if } n \le \tau \\
\tau,  &\mbox{ if } n > \tau \\
\end{array}
\right.$$
We use the classical notation for the indicator function $1_\mathcal{A}:$
$$1_\mathcal{A} = \left\{ \begin{array}{ll}
1, &\mbox{ if }  \mathcal{A} \mbox{ is true}\\
0,  & \mbox{ otherwise} \\
\end{array} \right.$$

Assume $Z_0 > C$, and for some $\epsilon>0$ and all $n \ge 0$, 
$$\Ec{\tilde{Z}_{n+1} }{\mathcal{F}_n}\le \tilde{Z}_n - \epsilon 1_{\tau >n}, \mbox{ almost surely.}$$
Then $\tau$ is almost surely finite and
$$\E{\tau} < \frac{Z_0}{\epsilon} < \infty.$$
\end{proposition}

\begin{proposition}[{ \cite[Theorem 3.2, Chapter 7]{doob1953}}] \label{prop:doob}
 Assume that for all $n$, $\Ec{Z_{n+1}-Z_n}{\mathcal{F}_n} \le 0$. Then for all $\lambda'>0$,
$$P[\sup_{n} Z_n \ge \lambda'] \le \frac{\E{Z_0}}{\lambda'}.$$
\end{proposition}

\end{document}